\documentclass[11pt, letterpaper]{article}

\usepackage[left]{lineno}

\usepackage[T1]{fontenc}
\usepackage{lmodern}
\usepackage{amsmath, amssymb, amsthm, dsfont, mathtools}
\usepackage{thmtools}
\usepackage{thm-restate}
\usepackage{xspace}
\usepackage{graphicx}
\usepackage{tikz}
\usetikzlibrary{decorations.pathmorphing}
\usetikzlibrary{arrows.meta,positioning,calc}
\usetikzlibrary{decorations.pathreplacing}
\usepackage[dvipsnames]{xcolor}
\definecolor{darkgreen}{RGB}{0,128,43}

\usepackage{enumitem}

\title{Temporal Cliques Admit Linear Spanners}
\author{Júlia Baligács\thanks{Email: \href{mailto:jbaligacs@gmail.com}{jbaligacs@gmail.com}. Funded by the European Union (ERC, CCOO, 101165139).
Views and opinions expressed are however those of the author only and do not necessarily reflect those of the European Union or the European Research Council. Neither the European Union nor the granting authority can be held responsible for them.}
\vspace{2mm}\\ University of Oxford}
\date{}

\usepackage[margin=1in]{geometry}

\usepackage{todonotes}
\usepackage{hyperref}
\usepackage[nameinlink]{cleveref}

\crefformat{lemma}{#2Lemma~#1#3}
\crefformat{theorem}{#2Theorem~#1#3}
\crefformat{proposition}{#2Proposition~#1#3}
\crefformat{remark}{#2Remark~#1#3}
\crefformat{equation}{#2(#1)#3}
\crefformat{corollary}{#2Corollary~#1#3}
\crefformat{example}{#2Example~#1#3}
\crefformat{definition}{#2Definition~#1#3}
\crefformat{figure}{#2Figure~#1#3}
\crefformat{observation}{#2Observation~#1#3}
\crefformat{section}{#2Section~#1#3}
\crefformat{question}{#2Question~#1#3}

\definecolor{LinkViolet}{hsb}{0.37,1,0.5}
\hypersetup{
  colorlinks=true,
  linkcolor=Violet,
  citecolor=Violet,
  urlcolor=green!40!black
}

\newcommand{\N}{\mathbb{N}}

\newcommand{\bigO}{\mathcal{O}}
\renewcommand{\subset}{\subseteq}

\DeclareMathOperator{\pos}{pos}
\DeclareMathOperator{\minmat}{N_{\mathrm{min}}}
\DeclareMathOperator{\maxmat}{N_{\mathrm{max}}}
\DeclareMathOperator{\ext}{Ext}
\DeclareMathOperator{\extind}{ind}
\renewcommand{\star}{\operatorname{Star}}
\newcommand{\minedges}{E_{\mathrm{min}}}
\newcommand{\maxedges}{E_{\mathrm{max}}}

\newcommand{\spiralarc}[6][]{%
  \begin{scope}[shift={#2}]
    \draw[->,thick,#1]
      plot[variable=\t,domain=0:1,samples=200,smooth]
        ({(#5 + (#6-#5)*\t) * cos(#3 + (#4-#3)*\t)},
         {(#5 + (#6-#5)*\t) * sin(#3 + (#4-#3)*\t)});
  \end{scope}%
}

\theoremstyle{definition}
\newtheorem{definition}{Definition}

\theoremstyle{plain}
\newtheorem{theorem}[definition]{Theorem}
\newtheorem{corollary}[definition]{Corollary}
\newtheorem{lemma}[definition]{Lemma}

\newtheorem{observation}[definition]{Observation}

\newtheorem{question}[definition]{Question}

\begin{document}

\maketitle
\thispagestyle{empty}

\begin{abstract}
A temporal graph is a graph in which every edge carries a non-empty set of time labels, and it is \emph{temporally connected} if for every two vertices $u$ and $v$, there exists a $u$-$v$-path with non-decreasing time labels. A \emph{spanner} is a subset of its edges preserving temporal connectivity.
Unlike static graphs, temporally connected graphs need not admit sparse spanners; nonetheless, minimizing spanner size is a central and widely studied problem.
A particularly intriguing question is whether temporal cliques admit spanners of linear size.
Despite considerable effort in recent years, the best known upper bound remained $O(n \log n)$.
We finally resolve this question, proving that every temporal clique on $n$ vertices admits a spanner of size~$7n$.
Moreover, such a spanner can be computed in polynomial time.
\end{abstract}

\section{Introduction}

A \emph{temporal graph} is a graph $G=(V,E,\lambda)$ in which each edge $e\in E$
carries a non-empty set of \emph{time labels} $\lambda(e)\subseteq\mathbb{R}$.
A \emph{temporal path} is a path $(v_1,\dots,v_k)$ whose edges admit non-decreasing time labels, that is, there exist~$l_i\in\lambda(\{v_i,v_{i+1}\})$ with $l_1\leq\cdots\leq l_{k-1}$, in which case $v_k$ is said to be \emph{temporally reachable} from $v_1$.
(While some works also consider the variant where labels need to be strictly increasing, we follow the classical non-decreasing convention.)
Then $G$ is \emph{temporally connected} if $u$ is temporally reachable from $v$ for every $u,v\in V$.

Temporal graphs model a wide range of real-world phenomena.
For example, in an information-dissemination setting, vertices represent agents and a time label on an edge records when its two endpoints communicate. 
A temporal $u$-$v$-path then certifies that agent $v$ can acquire information originating at $u$.
Beyond communication protocols, temporal graphs naturally encode transportation networks, the spread of infectious diseases, and dynamically evolving networks more generally.
Over the past decades, temporal graphs have also emerged as compelling theoretical objects and have attracted considerable attention (see surveys~\cite{Michail16, CasteigtsSurvey}).
A central theme is the search for temporal analogues of classical graph-theoretic results, particularly concerning connectivity.
Such analogues, however, frequently fail to hold outright or turn out to be significantly harder to establish~\cite{AkridaTVC18, Mertzios23,CasteigtsWaiting21,ZschocheJCSS20}.

Temporal reachability serves as a prime example of the latter.
Unlike in static graphs, temporal reachability is neither symmetric nor transitive, making connectivity a far more subtle notion.
Furthermore, properties of the underlying static graph offer essentially no guarantee of temporal connectivity: whenever $\{u,v\}\notin E$, one can assign label $2$ to all edges incident to $u$ and label $1$ to all edges incident to $v$, so that no temporal path from $u$ to $v$ exists.
Consequently, the only static graphs guaranteed to be temporally connected under any labeling are cliques.

A \emph{temporal spanner} of $G$ is a subset of edges $E'\subset E$ such that $(V,E',\lambda|_{E'})$ is temporally connected.
In the static setting, every connected graph on $n=|V|$ vertices admits a spanner of size $n-1$, namely a spanning tree, computable in polynomial time.
In the temporal setting, the situation is far more complex.
Kempe, Kleinberg, and Kumar~\cite{KempeKK00} asked whether every temporally connected graph admits a temporal spanner of size $\mathcal{O}(n)$, and immediately answered it negatively by constructing graphs in which every temporal spanner has size $\Omega(n\log n)$.
This was later strengthened by Axiotis and Fotakis~\cite{AxiotisF16}, who constructed graphs requiring spanners of size~$\Omega(n^2)$.
From a computational perspective, finding a temporal spanner of minimum size is~APX-hard~\cite{AxiotisF16, AkridaGMS17}.
In light of these results, a natural direction is to identify graph classes admitting sparse spanners.
A particularly compelling candidate is the class of temporal cliques:
as argued above, cliques are the only graphs guaranteeing temporal connectivity regardless of the labeling, making them a canonical setting in which to study temporal spanners.

\begin{question}[Casteigts, Peters, Schoeters, 2019]
\label{question}
Does every temporal clique admit a spanner of size~$\bigO(n)$?
\end{question}

Since it was first posed as an open problem in 2019~(conference version of~\cite{Casteigts21}), \cref{question} has been approached from many directions.
Even before the problem was explicitly posed, Akrida, Gąsieniec, Mertzios, and Spirakis~\cite{AkridaGMS17} gave the first improvement over the trivial bound, a spanner of size $\binom{n}{2}-\lfloor n/4\rfloor$ (still $\Theta(n^2)$), and showed that if each edge receives a single label chosen uniformly at random from an interval, then a spanner of size $\bigO(n\log n)$ exists asymptotically almost surely.
The most significant improvement is due to Casteigts, Peters, and Schoeters~\cite{Casteigts21}, who proved that every temporal clique admits a spanner of size~$\bigO(n\log n)$.
More recently, powerful new techniques~\cite{Angrick24,Carnevale25} pushed the linear regime within reach, but only for restricted subclasses such as the so-called edge-pivot cliques.
The problem has also been studied under additional constraints: bounding the length of the reachability paths~\cite{BiloESA22}, requiring the spanner to satisfy a robustness criterion~\cite{BiloDGLR24}, or building it from selfish agents located at the vertices~\cite{TemporalNetworkCreationGames23,
TemporalNetworkCreationGames25}, as well as for general temporal graphs~\cite{KempeKK00,AxiotisF16} and random temporal graphs~\cite{CasteigtsRandom}.
Despite this sustained effort, the~$\bigO(n\log n)$ bound remained the state of the art prior to our work.

\paragraph{Our results.}
We answer \cref{question} in the affirmative. 
Every temporal clique admits a spanner of \emph{linear} size.

\begin{restatable}{theorem}{mainthm}
\label{thm:mainthm}
Every temporal clique on $n$ vertices admits a temporal spanner of size $7n$.
\end{restatable}

This is asymptotically optimal, since every temporal spanner contains a spanning tree, i.e., at least $n-1$ edges, and, if edge labels are unique, at least $2n-4$ edges by a classical result in gossip theory~\cite{bumby1981}.

It is noteworthy that, despite the problem having resisted multiple independent attempts over several years, our proof of a linear bound is surprisingly short and self-contained: while it draws on ideas from prior work, it does not invoke any external result as a black box.
On its own it yields a spanner of size $14n$. The factor-of-two improvement to $7n$ is a minor refinement, obtained by invoking a result of~\cite{Carnevale25}, and is the only part of our work that relies on an external result.

To prove \cref{thm:mainthm}, we pass to the bipartite setting.
As observed in~\cite{Angrick24, Casteigts21}, \cref{question} is equivalent to its analogue for bi-cliques, which are often more convenient to work with. 
We make this equivalence precise in \cref{sec:preliminaries} and work with bi-cliques throughout.
Here, a \emph{temporal bi-clique} is a triple $G=(S,T,\lambda)$ where $S$ and $T$ are disjoint sets of \emph{sources} and \emph{targets}, the edge set is~$E:=\{\{s,t\}:s\in S,\,t\in T\}$, and $\lambda$ assigns each edge a non-empty set of time labels.
A \emph{spanner} is a subset of its edges through which every source reaches every target by a temporal path; paths from targets to sources are not required.

\begin{theorem}
\label{thm:mainthm_biclique}
Every temporal bi-clique $(S,T,\lambda)$ admits a temporal spanner of size
\[10\min(|S|,|T|)+2(|S|+|T|).
\]
\end{theorem}

We further point out that all spanners above can be computed in polynomial time.

\paragraph{Additional related work.}
Whereas we sparsify a given temporal graph, a related line of work asks how to construct the time labels in the first place. The usual goal is to make a static graph temporally connected with as few labels as possible~\cite{KlobasMFCS22}, sometimes under additional constraints~\cite{MertziosAlgorithmica19}. Motivated by epidemic control, the opposite objective has also been studied, i.e., reducing the number of vertices a single source can reach. This can be done by deleting edges~\cite{EnrightJCSS21}, reordering their activation times~\cite{Enright21Labels}, or delaying them~\cite{MolterMFCS21}. 

Temporal connectivity has also been studied on temporal versions of random graphs. A common model assigns each edge of an Erd\H{o}s--R\'enyi graph $G(n,p)$ a time label chosen uniformly at random from $[0,1]$. Here, sharp thresholds are known for various reachability properties~\cite{CasteigtsRandom}, including the emergence of a giant temporally connected component~\cite{Giantcomponents26}. Further properties have been studied in other models, such as the temporal diameter~\cite{BroutinKL24} and the size of the largest temporal clique~\cite{MertziosDeltaClique24,AtamanchukDL25}.

\paragraph{Outline and overview of techniques.}
In \cref{sec:preliminaries}, we collect and summarize the ideas from prior
work that we build on. Using these, we prove that \cref{thm:mainthm_biclique}
implies that temporal cliques contain spanners of linear size, and that we may assume that temporal bi-cliques satisfy an additional structural property known as
\emph{extremally matched}.

A generally useful technique throughout the paper is to find a set of edges solving the reachabilities of a large subgraph and then proceed with the unsolved part.
Formally, for two vertices~$s$ and~$t$, a set of edges $E'$ \emph{covers} $(s,t)$ if a temporal $s$-$t$-path exists using only edges of $E'$.
In \cref{sec:cuts_in_stars}, we introduce a natural set of edges called a \emph{simple star}, consisting of a star and a matching, which already covers half of all pairs $(s,t)\in S\times T$. We also define \emph{cut-crossing paths}, which can be appended to a simple star to cover additional pairs using only~$\bigO(n)$ additional edges.

In \cref{sec:intuition}, we present two results that, while not needed for the main proof, motivate the key ideas and reflect the proof discovery process.
We show that if $\Omega(n)$ cut-crossing paths can be appended to a simple star, the main result follows relatively easily.
This is, however, not the case in general: a construction from~\cite{Angrick24}, the \emph{shifted matching graph}, serves as a counterexample. 
Interestingly, the converse also holds: we prove that if \emph{no} cut-crossing paths can be appended to a simple star, then the graph is isomorphic to a shifted matching graph --- and these are known to admit linear spanners. 
In summary, both extreme cases, where $\Omega(n)$ paths are appendable or none are, admit linear spanners. It remains to handle the intermediate cases.
The main difficulty here lies in finding the correct notion, one that captures both extreme cases simultaneously.

Our main technical contribution is in  \cref{sec:extended_stars}. We introduce the carefully constructed \emph{extended stars}, which capture the reachabilities of a simple star together with a subset of its appendable paths. 
We prove that either there exist two extended stars, one centered at a source and one at a target (analogously to the spanner structure of a shifted matching graph), that together cover a sufficiently large subgraph, or there exists a simple star to which $\Omega(n)$ cut-crossing paths can be appended. 
Intuitively, the first condition holds when the graph is ``close'' to a shifted matching graph, and the second when it is ``far'' from one.
Therefore, we can find a linear-size set of edges covering a large subgraph, and recursively proceed on the uncovered part.
This allows us to complete our proof that temporal cliques contain spanners of linear size.

Finally, in \cref{sec:improved_factor}, we briefly show that the constant factor found in our construction can be easily halved by invoking a result from~\cite{Carnevale25}. In \cref{sec:shifted_matching_proofs}, we provide the missing proofs from \cref{sec:intuition}.
In \cref{sec:conclusion}, we recap the proof, observing that the steps can be carried out in polynomial time, and we conclude.

\section{Preliminaries}
\label{sec:preliminaries}

\subsection{Notation and conventions}
We collect a few simple observations that fix the setting and notation used throughout. 
We write a \emph{temporal clique} as~$G=(V,\lambda)$, omitting the edge set, which is always~$\binom{V}{2}$. 
First, observe that it suffices to prove Theorems~\ref{thm:mainthm} and~\ref{thm:mainthm_biclique} for temporal cliques and bi-cliques, respectively, in which every edge carries exactly one time label: when an edge has multiple labels, delete all but one; then any spanner of the modified graph is also a spanner of the original.
We therefore assume~$\lambda(e)\in\mathbb{R}$ for all edges $e$ throughout.
Moreover, since temporality of a path depends only on the relative ordering of its edge labels, we may assume without loss of generality that $\lambda(e)\in\mathbb{N}$ whenever convenient.

We may further assume that no two edges incident to the same vertex share a label: ties can be broken arbitrarily, and any spanner of the modified temporal graph is also a spanner of the original.
Under this assumption, each vertex~$v$ induces a total ordering on its neighbors, and we write~$\pos_v(u)=i$ if~$\{u,v\}$ is the~$(i+1)$-th smallest edge incident to~$v$ (so $\pos_v$ is~$0$-indexed).
We caution that~$\pos_v(u)=\pos_u(v)$ does \emph{not} hold in general.
In this notation, a path $(v_1,\dots,v_k)$ is temporal if and only if $\pos_{v_i}(v_{i-1})<\pos_{v_i}(v_{i+1})$ for every~$i\in\{2,\dots,k-1\}$.
Many of our arguments hinge on the neighbors of a vertex that occupy the smallest and largest positions in its induced ordering. 
For a vertex $v$, we define $\minmat(v)$ as the unique vertex~$w$ with $\pos_v(w)=0$, and $\maxmat(v)$ as the unique vertex $w$ maximizing $\pos_v(w)$.

To certify that~$v_k$ is temporally reachable from~$v_1$ in a spanner, we sometimes exhibit a \emph{temporal walk} rather than a temporal path, that is, a sequence~$W=(v_1,\dots,v_k)$ of vertices that may repeat, including consecutively. 
We call~$W$ \emph{temporal} if the sequence~$(u_1,\dots,u_{k'})$ obtained by collapsing each maximal run of identical consecutive vertices into a single occurrence satisfies $\pos_{u_i}(u_{i-1})\leq\pos_{u_i}(u_{i+1})$ for every~$i\in\{2,\dots,k'-1\}$. 
It is immediate that a temporal walk from~$v_1$ to~$v_k$ contains a temporal~\mbox{$v_1$-$v_k$}-path, so it serves as a valid certificate. 
We use walks because they let us avoid case distinctions when a vertex or edge may repeat.

\subsection{Temporal cliques and bi-cliques}

It was observed in~\cite{Angrick24} that the question of whether temporal cliques admit sparse spanners is in fact equivalent to the analogous question for bi-cliques.
To see the implication from cliques to bi-cliques, let $G = (S, T, \lambda)$ be a temporal bi-clique with edge labels in $[1,L]$.
We extend $G$ to a temporal clique by assigning label $L+1$ to all edges between two sources, and label $0$ to all edges between two targets (cf.~\cref{fig:bicliques}, left).
Observe that any temporal path from a source to a target cannot traverse any of the newly added edges.
Hence, every temporal spanner of the constructed clique contains a temporal spanner of the bi-clique~$G$.

For the reverse implication, the authors of \cite{Angrick24} showed how to associate to any temporal clique a temporal bi-clique whose spanners translate to spanners of the original graph (cf.~\cref{fig:bicliques}, right).
We employ this construction to prove that \cref{thm:mainthm_biclique} implies our main result that temporal cliques admit spanners of linear size.

\begin{figure}
\centering
\begin{tikzpicture}[scale=0.65, every node/.style={font=\small, circle, draw, fill, inner sep=0pt, minimum size=5pt}, baseline=(current bounding box.center)]
\foreach \x in {0,1,2,3}{
\node (s\x) at (0,-\x) {};
\node (t\x) at (2,-\x) {};
}
\foreach \x in {0,1,2,3}{
\foreach \y in {0,1,2,3}{
\draw[gray, thin] (s\x) to (t\y);
}}
\foreach \x in {0,...,3} {
  \foreach \y in {0,...,3} {
	\ifnum\x>\y
      \ifnum\x=\numexpr\y+1\relax
        \draw[red] (s\x) to [bend left=20] (s\y);
		\draw[blue] (t\x) to [bend right=20] (t\y);
      \else
        \draw[red] (s\x) to [bend left=40] (s\y);
		\draw[blue] (t\x) to [bend right=40] (t\y);
      \fi
    \fi
  }
}
\node at (-0.7,-1.5) [draw=none, fill=none, red, anchor=east] {$L+1$};
\node at (2.7,-1.5) [draw=none, fill=none, blue, anchor=west] {$0$};
\node at (1,-2.6) [draw=none, fill=none, gray, anchor=north] {$\in [1,L]$};
\node at (0,0.2) [draw=none, fill=none, anchor=south,rectangle] {$S$};
\node at (2,0.2) [draw=none, fill=none, anchor=south,rectangle] {$T$};
\end{tikzpicture}
\hspace{4cm}
\begin{tikzpicture}[scale=0.6, every node/.style={font=\small, circle, draw, fill, inner sep=0pt, minimum size=5pt}, baseline=(current bounding box.center)]
\node (u) at (0,0) {};
\node (v) at (2,0) {};
\node (w) at (1,1.2) {};
\draw[blue, thick] (u) to node[below, draw=none, fill=none, yshift=-1pt] {1} (v); 
\draw[orange, thick] (u) to node[above, draw=none, fill=none, xshift=-2pt] {2} (w); 
\draw[purple, thick] (v) to node[above, draw=none, fill=none, xshift=2pt] {3} (w); 
\node at (-0.2,-0.2) [draw=none, fill=none, anchor=east] {$u$};
\node at (2.2,-0.2) [draw=none, fill=none, anchor=west] {$w$};
\node at (1,1.4) [draw=none, fill=none, anchor=south] {$v$};
\end{tikzpicture}
\hspace{5mm}
\begin{tikzpicture}[scale=0.8, every node/.style={font=\small, circle, draw, fill, inner sep=0pt, minimum size=5pt}, baseline=(current bounding box.center)]
\node (us) at (0,2) {};
\node (vs) at (0,1) {};
\node (ws) at (0,0) {};
\node (ut) at (2,2) {};
\node (vt) at (2,1) {};
\node (wt) at (2,0) {};
\draw[thick, blue] (us) to  node [below, draw=none, fill=none, font=\scriptsize, pos=0.1, xshift=-1pt] {1} (wt);
\draw[thick, blue] (ws) to node [above, draw=none, fill=none, font=\scriptsize, pos=0.1, xshift=-1pt] {1} (ut);
\draw[thick, orange] (us) to node [above, draw=none, fill=none, font=\scriptsize, pos=0.35, yshift=0.5pt] {2} (vt);
\draw[thick, orange] (ut) to (vs);
\draw[thick, purple] (vs) to (wt);
\draw[thick, purple] (ws) to node [below, draw=none, fill=none, font=\scriptsize, pos=0.35, yshift=-0.5pt] {3} (vt);
\draw[thin,gray] (vs) to (vt);
\draw[thin,gray] (ws) to node [below, draw=none, fill=none, font=\scriptsize, pos=0.1,yshift=-1pt] {0}(wt);
\draw[thin,gray] (us) to node [above, draw=none, fill=none, font=\scriptsize, pos=0.1] {0} (ut);
\node at (-0.2,2) [draw=none,fill=none,anchor=east] {$u_s$};
\node at (-0.2,1) [draw=none,fill=none,anchor=east] {$v_s$};
\node at (-0.2,0) [draw=none,fill=none,anchor=east] {$w_s$};
\node at (2.2,2) [draw=none,fill=none,anchor=west] {$u_t$};
\node at (2.2,1) [draw=none,fill=none,anchor=west] {$v_t$};
\node at (2.2,0) [draw=none,fill=none,anchor=west] {$w_t$};
\end{tikzpicture}
\caption{On the left, the resulting clique contains a spanner of the underlying bi-clique.
On the right, we are given a temporal clique and the corresponding constructed bi-clique (\cref{obs:biclique_to_clique}).}
\label{fig:bicliques}
\end{figure}

\begin{restatable}{lemma}{lembiclique}
\label{obs:biclique_to_clique}
\cref{thm:mainthm_biclique} implies that every temporal clique on $n$ vertices admits a spanner of size~$14n$.
\end{restatable}

\begin{proof}
  Given a temporal clique $G = (V, \lambda)$ on $n$ vertices with strictly positive edge
  labels, we construct a temporal bi-clique $G' = (S, T, \lambda')$ as follows
  (cf.~\cref{fig:bicliques}, right). For every $v \in V$, introduce a copy
  $v_S \in S$ and a copy $v_T \in T$. Set $\lambda'(\{v_S, v_T\}) := 0$ for all~$v$
  and $\lambda'(\{v_S, u_T\}) := \lambda(\{v, u\})$ for all $v \neq u$.

By \cref{thm:mainthm_biclique}, $G'$ admits a temporal spanner $E'$ of size $10\min(|S|,|T|) + 2(|S|+|T|) = 14n$.
Define $E^*\subset V^{(2)}$ to be its projection, that is,
 $
    E^* := \bigl\{\{u,v\} : \{u_S, v_T\} \in E', u \neq v\bigr\}.
 $
Then $|E^*| \leq |E'| \leq 14n$, and it remains to show that $E^*$ is a temporal spanner of $G$. 
Indeed, for every~$u, v \in V$, there exists a temporal $u_S$-$v_T$-path $P'$ in $E'$. 
Project~$P'$ vertex by vertex, mapping~$x_S$ and~$x_T$ both to $x$.
Edges of the form $\{x_S,x_T\}$ become consecutive repetitions and are collapsed (this can occur only in a prefix of $P'$, since all edge labels of $G$ are strictly positive). 
Every remaining edge corresponds to an edge of $E^*$, and the edge labels are preserved; hence the projection is a temporal walk in $G$.
\end{proof}

For the remainder of the paper, we therefore focus on temporal bi-cliques.
In some scenarios, they can be easier to work with because they separate the roles of sources and targets. 
In a clique, every vertex plays both roles simultaneously.

\subsection{Dismountability and extremally matched bi-cliques}

\newcommand{\minneighbor}{N_\mathrm{min}}
\newcommand{\maxneighbor}{N_\mathrm{max}}

In \cite{Casteigts21}, the authors introduced the concept of \emph{dismountability},\footnote{
  The exact definitions of \emph{dismountable} and \emph{extremally matched}
  vary slightly throughout the literature.
  The definitions given here are tailored to our purposes.}
 which was adapted to bi-cliques in \cite{Angrick24} and revisited in \cite{Carnevale25}.
It is a simple yet powerful tool based on the following idea (cf.~\cref{fig:dismountability}, left):
Given a temporal bi-clique $(S,T,\lambda)$, let $s\in S$, $t:=\minneighbor(s)$, and suppose there exists~$s'\in S$ with $\pos_t(s')< \pos_t(s)$.
Then the concatenation of the path~$(s',t,s)$ with any temporal path from~$s$ yields a temporal walk.
Therefore, if a set of edges $E'$ covers $(s,T)$ (i.e., covers $(s,v)$ for every~$v\in T$), then
$E'\cup \{\{s,t\}, \{s',t\}\}$ also covers $(s',T)$.
This means that we can delete $s'$ from the graph, find a spanner of the remaining graph, and then reintroduce $s'$ together with the edges~$\{s,t\}$ and $\{s',t\}$.
Since we add only a constant number of edges (two) per deleted source, this is compatible with the goal of finding linear-sized spanners.
In this case, we say that~$s'$ is \emph{dismountable}
\emph{via $(s,t)$}.

Note that, if no dismountable sources exist, then for every $s \in S$ we have $\minneighbor(\minneighbor(s))=s$.
In particular, no two sources share the same smallest neighbor, and hence $|S|\leq |T|$.
The same concept applies to targets (cf.~\cref{fig:dismountability}, right):
given $t,t'\in T$ with $s:=\maxneighbor(t)$ and $\pos_s(t')> \pos_s(t)$,
one can extend any temporal walk ending in $t$ by appending the path $(t,s,t')$, allowing us to delete~$t'$, find a spanner of the remaining graph, and reintroduce $t'$ together with the edges $\{s,t\}$ and $\{s,t'\}$.
We then say that $t'$ is \emph{dismountable via $(s,t)$}.
Analogously, this yields a graph in which, for every~$t \in T$ we have $\maxneighbor(\maxneighbor(t))=t$, and in particular $|T|\leq|S|$.

It follows that, if there are neither dismountable sources nor dismountable targets, the bi-clique satisfies the following properties, and we call it \emph{extremally matched}.
\begin{itemize}
  \item $\minedges:=\{\{s,\minneighbor(s)\}:s \in S\}=
        \{\{t,\minneighbor(t)\}:t \in T\}$ is a perfect matching.
  \item $\maxedges:=\{\{s,\maxneighbor(s)\}:s \in S\}=
        \{\{t,\maxneighbor(t)\}:t \in T\}$ is a perfect matching.
\end{itemize}
In particular, extremally matched bi-cliques satisfy $|S|=|T|$.
We summarize these findings in the following result, which was already noted (albeit in a slightly different form) in~\cite{Angrick24}.
Here, given $S'\subset S, T'\subset T$, we denote by $G[S',T']:=(S',T',\lambda|_{\{\{s,t\}\,:\,s\in S',\,t\in T'\}})$ the induced temporal subgraph.

\begin{figure}
\centering
 \begin{tikzpicture}[scale=0.45, every node/.style={font=\small, circle, draw, fill, inner sep=0pt, minimum size=5pt}]
\node (s1) at (0,2) {};
\node (t) at (3,1) {};
\node (s) at (0,0) {};
\draw[red, thick] (s1) to (t) to (s);
\foreach \a in {0,10,20,30,40} {
    \draw (s) -- ++(-\a:2.5);
  }
\node at (-0.6,2) [draw=none, fill=none] {$s'$};
\node at (-0.6,0) [draw=none, fill=none] {$s$};
\node at (3.4,1) [draw=none, fill=none, anchor=west] {$t=\minneighbor(s)$};
\draw[thick, ->] (2.5, 1.5) to [bend right=30] (2.5,0.4);
\draw[thick, ->] (0.5, 0.5) to [bend left=30] (0.4,-0.9);
\node at (0,-1.5) [draw=none, fill=none] {};
\end{tikzpicture}
\hspace{1cm}
\begin{tikzpicture}[scale=0.45, every node/.style={font=\small, circle, draw, fill, inner sep=0pt, minimum size=5pt}]
\node (s1) at (0,-2) {};
\node (t) at (-3,-1) {};
\node (s) at (0,0) {};
\draw[red, thick] (s1) to (t) to (s);
\foreach \a in {0,10,20,30,40} {
    \draw (s) -- ++(180-\a:2.5);
  }
\node at (0.6,-2) [draw=none, fill=none] {$t'$};
\node at (0.6,0) [draw=none, fill=none] {$t$};
\node at (-3.4,-1) [draw=none, fill=none, anchor=east] {$s=\maxneighbor(t)$};

\draw[thick, <-] (-2.5, -1.6) to [bend right=30] (-2.5,-0.5);
\draw[thick, <-] (-0.5, -0.6) to [bend left=30] (-0.4,0.8);
\node at (0,-3.5) [draw=none,fill=none] {};
\end{tikzpicture}
\caption{On the left, $s'$ is dismountable via $(s,t)$. On the right, $t'$ is dismountable via $(s,t)$.}
\label{fig:dismountability}
 \end{figure}

\begin{restatable}[Dismountability]{lemma}{lemdismount}
\label{lem:dismountable}
Given a temporal bi-clique $G=(S,T,\lambda)$, there exist $S'\subseteq S$, $T'\subseteq T$, and a set $E'$ of at most $2(|S|+|T|-|S'|-|T'|)$ edges such that $G[S',T']$ is extremally matched and, if $E^*$ is a temporal spanner of $G[S',T']$, then $E^*\cup E'$ is a temporal spanner of $G$.
\end{restatable}

\begin{proof}
Consider the algorithm that initializes $S'\coloneqq S$, $T'\coloneqq T$, $E'\coloneqq\emptyset$ and repeatedly adjusts these sets as follows.
Throughout, dismountability refers to dismountability in the current graph~$G[S',T']$.
\begin{itemize}
  \item If $s'\in S'$ is dismountable via $(s,t)$, remove $s'$ from $S'$
        and add $\{s,t\}$ and $\{s',t\}$ to $E'$.
  \item If $t'\in T'$ is dismountable via $(s,t)$, remove $t'$ from $T'$
        and add $\{s,t\}$ and $\{s,t'\}$ to $E'$.
\end{itemize}
The algorithm terminates once no dismountable vertices remain, and the resulting $S'$, $T'$, $E'$ satisfy the claimed properties by construction.
\end{proof}

This allows us to restrict our attention to extremally matched bi-cliques.
More precisely, we establish next that it suffices to prove the following.

\begin{restatable}{theorem}{mainthmextremallymatched}
\label{thm:mainthm_extremally_matched}
Every extremally matched temporal bi-clique $(S,T,\lambda)$ with $n:=|S|=|T|$ contains a temporal spanner of size $14n$.
\end{restatable}

\begin{proof}[Proof that \cref{thm:mainthm_extremally_matched} implies \cref{thm:mainthm_biclique}]
Given a temporal bi-clique $G=(S,T,\lambda)$, by \cref{lem:dismountable} there exist $S'\subseteq S$ and $T'\subseteq T$ such that $G[S',T']$ is extremally matched, and a set $E'$ of at most $2(|S|+|T|-|S'|-|T'|)$ edges such that any temporal spanner of $G[S',T']$ can be extended to a temporal spanner of $G$ by adding the edges in $E'$.
By \cref{thm:mainthm_extremally_matched}, $G[S',T']$ admits a temporal spanner of size $14|S'|$. 
Adding the edges in $E'$ yields a temporal spanner of $G$ of size
\begin{equation*}
  14|S'| + 2(|S|+|T|-|S'|-|T'|)
  \overset{|S'|=|T'|}= 10|S'| + 2(|S|+|T|)
  \leq 10\min(|S|,|T|) + 2(|S|+|T|).\qedhere
\end{equation*}
\end{proof}

In the next two sections, we focus on proving \cref{thm:mainthm_extremally_matched}, which, as shown in this section, implies our
main results Theorem~\ref{thm:mainthm_biclique} and, by \cref{obs:biclique_to_clique}, the linear-size bound for temporal cliques.

\section{Cuts in stars}
\label{sec:cuts_in_stars}

Throughout this section, let $G=(S,T,\lambda)$ be an extremally matched temporal bi-clique and assume as usual that no two incident edges have the same label so that $\pos_s(t)$ and $\pos_t(s)$ are well-defined for every $s\in S, t\in T$.
Many of the statements in this and the next section come in two analogous forms, one from the perspective of~$S$ and one from that of~$T$. 
We label such pairs of statements~(a) and~(b), and use~(i), (ii) for parts that are not analogous.
First, we introduce a canonical way to label the vertices of $G$ that we use throughout this and the next section.

\begin{definition}
\label{def:induced_labeling}
\leavevmode
\begin{enumerate}[topsep=2pt, label=(\alph*)]
\item 
\label{def:labeling_source}
For $s^* \in S$, the \emph{$s^*$-ordered labeling} (cf.~\cref{fig:extended_stars}, left) is the vertex labeling~$T = \{t_0, \dots, t_{n-1}\}$ and $S = \{s_0, \dots, s_{n-1}\}$  
where the $t_i$ are ordered such that
$\pos_{s^*}(t_i) < \pos_{s^*}(t_{i+1})$ for all $i\leq n-2$
and $s_i = \minmat(t_i)$ for all $i \leq n-1$.
\item 
\label{def:labeling_target}
For $t^* \in T$, the \emph{$t^*$-ordered labeling} (cf.~\cref{fig:extended_stars}, right) is the vertex labeling $T = \{t_0, \dots, t_{n-1}\}$ and $S = \{s_0, \dots, s_{n-1}\}$ 
where the $s_i$ are ordered such that 
$\pos_{t^*}(s_i)>\pos_{t^*}(s_{i+1})$ for all $i\leq n-2$ and
$t_i=\maxmat(s_i)$ for all $i\leq n-1$.
\end{enumerate}
\end{definition}

Note that, since the bi-clique is extremally matched, we have in the $s^*$-ordered labeling that $s_0=\minmat(t_0)=\minmat(\minmat(s^*))=s^*$.
Next, we define a simple set of edges already covering a large set of vertices.

\begin{definition}
\label{def:cut_crossing}
Let $s^*\in S$ and consider the $s^*$-ordered labeling.
\begin{enumerate}[label=(\roman*), topsep=4pt]
\item The \emph{simple star from $s^*$}, denoted $\star(s^*)$, consists of $\minedges$ and all edges incident to~$s^*$.
\item Let $k\in \{0,\dots, n-2\}$. We say that a source $s_i$ is \emph{$m$-hop-$k$-cut-crossing for $s^*$} (cf.~\cref{fig:cut_star}) if $i>k$ and there exist $j\leq k$ and a temporal path of length at most $m+1$ that begins in $s_i$ and ends with the edge $\{t_j,s^*\}$.
\end{enumerate}
\end{definition}

\begin{figure}
\centering
 \begin{tikzpicture}[scale=0.7, every node/.style={font=\small, circle, draw, fill, inner sep=0pt, minimum size=5pt}]
\foreach \i in {0,...,6}{
  \node (s\i) at (0,-\i) {};
\node (t\i) at (2,-\i) {};
\draw (s0) to (t\i);
\draw[dashed, gray] (s\i) to (t\i);
}
\node at (-0.3,0) [draw=none, fill=none, anchor=east,rectangle] {$s^*=s_0$};
\node at (-0.3,-5) [draw=none, fill=none, anchor=east] {$s_i$};
\node at (-0.3,-3) [draw=none, fill=none, anchor=east] {$s_k$};
\node at (-0.3,-6) [draw=none, fill=none, anchor=east] {$s_{n-1}$};
\node at (2.3,0) [draw=none, fill=none, anchor=west,rectangle] {$t_0$};
\node at (2.3,-2) [draw=none, fill=none, anchor=west] {$t_j$};
\node at (2.3,-4) [draw=none, fill=none, anchor=west] {$t_{k+1}$};
\node at (2.3,-6) [draw=none, fill=none, anchor=west] {$t_{n-1}$};
\node at (0.5,-5.83)  [draw=none, fill=none, font=\scriptsize, color=gray] {min};
\draw[thick, blue] (-2,-3.5) to (4,-3.5);
\node at (4.2, -3.5) [draw=none, fill=none, anchor=west, blue] {$k$-cut};
\node at (-2.3, -1.5) [draw=none, fill=none, anchor=east, blue] {$S'$};
\node at (3.7, -4.5) [draw=none, fill=none, anchor=west, blue] {$T'$};
\draw[thick, darkgreen] (s5) to (t2);
\draw[->,thick] (10:0.5) arc[start angle=10,end angle=-100,radius=0.5];
\draw[thick, darkgreen, ->] (1.9, -2.4) to [bend left=40] (1.9, -1.6);
\draw[thick, blue, decorate,decoration={brace,amplitude=6pt}]
    (-2,-3.2) -- (-2,0.2);
\draw[thick, blue, decorate,decoration={brace,amplitude=6pt}]
    (3.3,-2.8) -- (3.3,-6.2);
\end{tikzpicture}
\caption{A 1-hop-$k$-cut-crossing source $s_i$ (\cref{def:cut_crossing}).}
\label{fig:cut_star}
\end{figure}

The motivation for the definition above is as follows (cf.~\cref{fig:cut_star}): First, observe that, for every~$j\leq i$, $\star(s^*)$ covers the pair $(s_j, t_i)$ via the temporal walk $(s_j, t_j, s^*,t_i)$.
In particular, given some $k\in \{0,\dots, n-2\}$, we can set $S':=\{s_0, \dots, s_k\}$, $T':=\{t_{k}, \dots, t_{n-1}\}$, and obtain that~$\star(s^*)$ covers $(S',T')$, that is, it covers $(s,t)$ for every $s\in S', t\in T'$.

Further, observe that, if a source $s_i$ is $m$-hop-$k$-cut-crossing via the path $P$, then $\star(s^*)\cup P$ also covers~$(s_i,T')$ via the temporal walks given by $P\circ (s^*,t)$ for any $t \in T'$ (cf.~\cref{fig:cut_star}), where~$W_1 \circ W_2$ denotes the concatenation of the walks $W_1,W_2$.
In other words, $s_i$ being $m$-hop-$k$-cut-crossing means that we can add $s_i$ to $S'$ by additionally adding~$m$ edges to $\star(s^*)$ (those $m$ edges are the edges of $P$, where the last edge of $P$ is already contained in $\star(s^*)$).

We remark that the definition could be easily extended to incorporate more cases, such as targets crossing the cut, and the analogous case where the star center is a target. 
But since our proof does not need that generality, we state only the one-sided version here for brevity.

The following observation summarizes the above discussion.

\begin{observation}
\label{obs:cut_crossing}
Let $s^*\in S$, $m\in \N$, $k\in \{0,\dots, n-2\}$, and consider the $s^*$-ordered labeling.
Let $S_\mathrm{cross}$ denote the set of sources that are  $m$-hop-$k$-cut-crossing, $T':=\{t_{k}, \dots, t_{n-1}\}$, and~$S':=\{s_0, \ldots, s_k\}\cup S_\mathrm{cross}$.
Then there exists a set of $|S_\mathrm{cross}|\cdot m$ edges $E'$ such that $E'\cup \star(s^*)$ covers~$(S',T')$.
\end{observation}

\subsection{Intuition for the proof of {\cref{thm:mainthm_extremally_matched}}}
\label{sec:intuition}

In this subsection, we explain the intuition behind our proof of \cref{thm:mainthm_extremally_matched}.
Its two results are not needed, but they explain the structure adopted in the next section.
We state them here and defer their proofs to \cref{sec:shifted_matching_proofs}.
Readers seeking only a direct proof of \cref{thm:mainthm_extremally_matched} may skip to \cref{sec:extended_stars}.

The next lemma shows that the existence of linear-size spanners follows if we can find a linear number of edges that covers a sufficiently large subgraph, and it can be proven relatively easily by recursing on the uncovered parts and applying dismountability (\cref{lem:dismountable}).

\begin{restatable}{lemma}{recursionlemma}
\label{lem:recursion}
Suppose there exist constants $C\geq 1$, $\delta \in (0,1]$ such that, for every extremally matched temporal bi-clique $G=(S,T,\lambda)$ of size $n:=|S|$, there exist $S'\subset S$, $T'\subset T$ and a set $E'$ of at most~$Cn$ edges such that
\begin{enumerate}[label=(\roman*)]
    \item $E'$ covers $(S',T')$,
    \item $|S'|+|T'|\geq (1+\delta) n$.
\end{enumerate}
Then every extremally matched temporal bi-clique contains a spanner of size $\bigO(n)$.
\end{restatable}

\begin{proof}[Proof sketch]
It remains to cover $(S\setminus S',T)$ and $(S,T\setminus T')$. By dismountability (\cref{lem:dismountable}), each reduces to an extremally matched instance of size at most $|S\setminus S'|$, respectively $|T\setminus T'|$, at the cost of linearly many additional edges. The total size of the remaining instances is therefore $|S\setminus S'|+|T\setminus T'|=2n-|S'|-|T'|\leq (1-\delta)n$.
With this, a straightforward recursion yields a spanner of size $\bigO(n)$, where the constant factor hidden in the $\bigO$-notation depends on $C$ and $\delta$.
\end{proof}

An immediate consequence of \cref{obs:cut_crossing} is the following: if there exists a constant $m\in \mathbb{N}$ such that every extremally matched temporal bi-clique $(S,T,\lambda)$ admits some $s^*\in S$ and $k\in \{0, \dots, n-2\}$ with $\Omega(n)$ many $m$-hop-$k$-cut-crossing sources, then the assumptions of \cref{lem:recursion} are satisfied, and a linear bound follows.
One might hope to prove this in general, but this fails: we prove that the \emph{shifted matching graphs} provide a counterexample.
The shifted matching graph of size~$n$ is defined by (cf.~\cref{fig:shifted_matching}) $S:=\{s_0, \dots, s_{n-1}\}$, $T:=\{t_0, \dots, t_{n-1}\}$, and ${\lambda(\{s_i,t_j\}):=j-i \bmod n}$.
They were introduced in \cite{Angrick24} alongside a notion called \emph{edge pivotability}, related to our definition of crossing cuts in stars. There, the authors showed that edge-pivotable temporal bi-cliques admit linear spanners, while the shifted matching graphs witness that not every temporal bi-clique is edge-pivotable. It is therefore unsurprising that they also fail to have~$\Omega(n)$ crossing sources.

We prove something stronger, however: having no 1-hop-$k$-cut-crossing sources for any $s^*$ and~$k$ characterizes the shifted matching graph up to isomorphism. Here, two temporal bi-cliques $(S_1,T_1,\lambda_1)$ and $(S_2,T_2,\lambda_2)$ are \emph{isomorphic} if there exist bijections $\Psi_S \colon S_1 \to S_2$ and $\Psi_T\colon T_1 \to T_2$ such that $\pos_{s}(t)=\pos_{\Psi_S(s)}(\Psi_T(t))$ and $\pos_{t}(s)=\pos_{\Psi_T(t)}(\Psi_S(s))$ for all $s\in S_1$, $t\in T_1$.

\begin{restatable}{proposition}{shiftedmatching}
Let $G=(S,T,\lambda)$ be an extremally matched temporal bi-clique with $n:=|S|$. Then~$G$ is isomorphic to the shifted matching graph of size $n$ if and only if, for every $s^*\in S$ and every $k\in\{0,\dots,n-2\}$, there is no 1-hop-$k$-cut-crossing source.
\end{restatable}

However, as noted in \cite{Angrick24}, the shifted matching graph admits a linear spanner (see \cref{fig:shifted_matching}, right): let $E'$ consist of $\minedges$, $\maxedges$, all edges incident to $s_0$, and all edges incident to $t_0$. For~$i\leq j$, the pair $(s_i,t_j)$ is covered by the temporal walk $(s_i,t_i,s_0,t_j)$; for $i > j$, by $(s_i,t_0,s_{j+1},t_j)$.

To summarize, if an extremally matched temporal bi-clique has a simple star and a $k$-cut with~$\Omega(n)$ crossing sources, \cref{obs:cut_crossing} together with \cref{lem:recursion} yields a linear spanner. If no~$k$-cut has any crossing sources for any star, the graph is isomorphic to a shifted matching graph, which also admits a linear spanner. It remains to handle the intermediate cases.
To this end, we introduce in the next section the notion of \emph{extended stars}. These are simple stars augmented by a linear number of edges that capture many cut-crossing paths. We then show that either some simple star has many cut-crossing sources, or there exist two extended stars, one centered at a source and one at a target, similarly to the spanner structure of the shifted matching graph, that together cover a sufficiently large subgraph.

\section{Extended stars}
\label{sec:extended_stars}

This section contains our main technical contribution. We introduce a carefully constructed structure of $\bigO(n)$ edges called an \emph{extended star} and prove that either two such structures together cover a sufficiently large subset of vertices, or some simple star has sufficiently many $3$-hop-$k$-cut-crossing sources. This then allows us to conclude that temporal cliques admit spanners of linear size.

Throughout this section, let $G=(S,T,\lambda)$ be an extremally matched temporal bi-clique of size~$n:=|S|=|T|$, and recall that no two incident edges share a label, so that $\pos_s(t)$ and~$\pos_t(s)$ are well-defined for every $s\in S, t\in T$.
We begin by defining extended stars.

\begin{figure}
\centering
\begin{tikzpicture}[scale=0.75, every node/.style={font=\small, circle, draw, fill, inner sep=0pt, minimum size=5pt}]
\foreach \i in {0,...,6}{
  \node (s\i) at (0,-\i) {};
\node (t\i) at (2,-\i) {};
\draw (s0) to (t\i);
\draw[dashed, gray] (s\i) to (t\i);
}
\node at (-0.3,0) [draw=none, fill=none, anchor=east, minimum size=0pt,rectangle] {$s^*=s_0$};
\node at (-0.3,-1) [draw=none, fill=none, anchor=east] {$s_1$};
\node at (-0.3,-2) [draw=none, fill=none, anchor=east] {$s_2$};
\node at (-0.3,-6) [draw=none, fill=none, anchor=east,rectangle] {$s_{n-1}$};
\node at (2.3,0) [draw=none, fill=none, anchor=west, minimum size=0pt,rectangle] {$t_0=\minmat(s^*)$};
\node at (2.3,-1) [draw=none, fill=none, anchor=west] {$t_1$};
\node at (2.3,-2) [draw=none, fill=none, anchor=west] {$t_2$};
\node at (2.3,-6) [draw=none, fill=none, anchor=west,rectangle] {$t_{n-1}=\maxmat(s^*)$};
\draw[thick, red] (s2) to (t4);
\node at (2.1,-4) [draw=none, fill=none, anchor=west, align=center,text width=2cm, font=\footnotesize] {extend index\\ of $s_2$};
\draw[->,thick] (10:0.5) arc[start angle=10,end angle=-100,radius=0.5];
\draw[thick, red,<-] (2,-4) ++(167:0.35) arc[start angle=167,end angle=90,radius=0.35];
\node at (0.5,-5.83)  [draw=none, fill=none, font=\scriptsize, color=gray] {min};
\end{tikzpicture}
\hspace{17mm}
\begin{tikzpicture}[scale=0.75, every node/.style={font=\small, circle, draw, fill, inner sep=0pt, minimum size=5pt}]
\foreach \i in {0,...,6}{
  \node (s\i) at (0,-\i) {};
\node (t\i) at (2,-\i) {};
\draw (t0) to (s\i);
\draw[dashed, gray] (s\i) to (t\i);
}
\node at (-0.3,0) [draw=none, fill=none, anchor=east, minimum size=0pt,rectangle] {$\maxmat(t^*)=s_0$};
\node at (-0.3,-1) [draw=none, fill=none, anchor=east] {$s_1$};
\node at (-0.3,-2) [draw=none, fill=none, anchor=east] {$s_2$};
\node at (-0.3,-6) [draw=none, fill=none, anchor=east,rectangle] {$\minmat(t^*)=s_{n-1}$};
\node at (2.3,0) [draw=none, fill=none, anchor=west, minimum size=0pt,rectangle] {$t_0=t^*$};
\node at (2.3,-1) [draw=none, fill=none, anchor=west] {$t_1$};
\node at (2.3,-2) [draw=none, fill=none, anchor=west] {$t_2$};
\node at (2.3,-6) [draw=none, fill=none, anchor=west,rectangle] {$t_{n-1}$};
\draw[thick, red] (t2) to (s5);
\node at (-0.1,-5) [draw=none, fill=none, anchor=east, align=center,text width=2cm, font=\footnotesize] {extend index\\ of $t_2$};
\draw[->,thick] (2,0) ++(-80:0.5) arc[start angle=-80,end angle=-200,radius=0.5];
\draw[thick, red,->] (0,-5) ++(40:0.35) arc[start angle=40,end angle=100,radius=0.35];
\node at (1.5,-5.85)  [draw=none, fill=none, font=\scriptsize, color=gray] {max};
\end{tikzpicture}
\caption{The left subfigure illustrates the labeling induced by $s^*$ (\cref{def:induced_labeling}~\ref{def:labeling_source}). The dashed edges denote the minimum label matching $\minedges$. The extend index of $s_2$ (\cref{def:extended_star}~\ref{def:extend_index_source}) is the index of the bottommost target satisfying the ordering indicated by the red arrow.\\
The right subfigure illustrates the labeling induced by $t^*$ (\cref{def:induced_labeling}~\ref{def:labeling_target}). The dashed edges denote the maximum label matching $\maxedges$. The extend index of $t_2$ (\cref{def:extended_star}~\ref{def:extend_index_target}) is the index of the bottommost source satisfying the ordering indicated by the red arrow.}
\label{fig:extended_stars}
\end{figure}

\begin{definition}
\label{def:extended_star}
Let $s^*\in S$ and $t^*\in T$. For $x\in\{s^*,t^*\}$, the \emph{extended star from $x$}, denoted $\ext(x)$, consists of $\minedges$, $\maxedges$, all edges incident to $x$, and the following additional edges.
\begin{enumerate}[label=(\alph*),topsep=4pt]
\item 
\label{def:extend_index_source}
If $x=s^*$, consider the labeling induced by $s^*$ and, for each $s\in S\setminus\{s^*\}$, define its \emph{extend index} (cf.~\cref{fig:extended_stars}, left) as
\[
    \extind (s) \;:=\; \max\bigl\{\,l \in \{0,\dots, n-1\}: \pos_{t_l}(s) > \pos_{t_l}(s^*)\,\bigr\}
\]
and add the edge $\{s,\,t_{\extind(s)}\}$ to~$\ext(s^*)$.
\item 
\label{def:extend_index_target}
If $x=t^*$, consider the labeling induced by $t^*$ and, for each $t\in T\setminus\{t^*\}$, define its \emph{extend index}~(cf.~\cref{fig:extended_stars}, right) as
\[
    \extind (t) \;:=\; \max\bigl\{\,l \in \{0,\dots, n-1\}: \pos_{s_l}(t) < \pos_{s_l}(t^*)\,\bigr\}
\]
and add the edge $\{s_{\extind(t)},t\}$ to~$\ext(t^*)$.
\end{enumerate}
\end{definition}

Observe here that the extend index is well-defined because the set over which the maximum is taken always contains $l=0$ and is therefore nonempty.

The following lemma identifies a large class of source-target pairs covered by extended stars, and conveys the main intuition behind their definition.

\begin{lemma}
\label{obs:extend_index_cover}
\leavevmode
\begin{enumerate}[topsep=2pt, label=(\alph*)]
\item Let $s^*\in S$, and consider the labeling induced by $s^*$. 
Let $t\in T\setminus \{t_{n-1}\}$ and $i:=\extind(\maxmat(t))$.
Then~$\ext(s^*)$ covers $(s_j,t)$ for all $j\le i$.
\label{obs:extend_index_cover_source}
\item Let $t^*\in T$, and consider the labeling induced by $t^*$. 
Let $s\in S\setminus\{s_{n-1}\}$ and $i:=\extind(\minmat(s))$.
Then~$\ext(t^*)$ covers $(s,t_j)$ for all $j\le i$.
\label{obs:extend_index_cover_target}
\end{enumerate}
\end{lemma}

\begin{proof}
For part \ref{obs:extend_index_cover_source}, observe that $\maxmat(t)$ having extend index $i$ implies, by definition, that the path $(s^*,t_i, \maxmat(t))$ is temporal. 
Given $j\leq i$, we obtain that the walk 
\[
(s_j, t_j, s^*, t_i, \maxmat(t),t)
\] 
is temporal (cf.~\cref{fig:extended_stars}, left), where we used that $\minmat(t_j)=s_j$, and that the bi-clique is extremally matched, so the edge~$\{\maxmat(t),t\}$ is maximal at $\maxmat(t)$.
 Moreover, the walk is contained in~$\ext(s^*)$, so that~$\ext(s^*)$ covers~$(s_j, t)$.

Part \ref{obs:extend_index_cover_target} follows analogously from the fact that the walk 
\(
(s, \minmat(s),s_i,t^*,s_j, t_j)
\)
is temporal for every $j\leq i$ (cf.~\cref{fig:extended_stars}, right).
\end{proof}

Conversely, if an extended star does \emph{not} cover a pair, we can deduce the following.

\begin{lemma}
\label{lem:no_cover_consequence}
Let $s,s^*\in S$ and $t, t^* \in T$.
\begin{enumerate}[label=(\alph*), topsep=4pt]
\item If $\ext(s^*)$ does \emph{not} cover $(s,t)$, then $\pos_{\minmat(s)}(\maxmat(t))\leq \pos_{\minmat(s)}(s^*)$.
\label{lem:no_cover_consequence_source}
\item If $\ext(t^*)$ does \emph{not} cover $(s,t)$, then $\pos_{\maxmat(t)}(t^*)\leq \pos_{\maxmat(t)}(\minmat(s))$.
\label{lem:no_cover_consequence_target}
\end{enumerate}
\end{lemma}

\begin{proof}
For part \ref{lem:no_cover_consequence_source}, consider the labeling induced by $s^*$ and let $j$ be chosen such that $s=s_j$.
If~$\maxmat(t)=s^*$, the desired inequality holds with equality, so we can assume $\maxmat(t)\neq s^*$. (One can even observe that $\maxmat(t)=s^*$ is a contradiction to $(s,t)$ not being covered.)
If~$\maxmat(t)$ has an extend index that is greater than or equal to $j$,  then by \cref{obs:extend_index_cover}~\ref{obs:extend_index_cover_source}, $\ext(s^*)$ covers the pair~$(s_j,t)=(s,t)$, which is a contradiction.
Therefore, the extend index of $\maxmat(t)$ is strictly smaller than $j$. 
By \cref{def:extended_star}~\ref{def:extend_index_source} (with $\maxmat(t)\neq s^*$), this means that
\begin{equation*}
\pos_{t_l}(\maxmat(t))\leq \pos_{t_l}(s^*) \text{ for every  $l\geq j$}.
\end{equation*}
In particular, this holds for $l=j$, i.e., when $t_l=t_j=\minmat(s_j)=\minmat(s)$.

For part \ref{lem:no_cover_consequence_target}, consider the labeling induced by $t^*$ and let $j$ be such that $t=t_j$. 
Analogously to part~\ref{lem:no_cover_consequence_source}, the case $\minmat(s)=t^*$ is trivial, so we can assume $\minmat(s)\neq t^*$.
By \cref{obs:extend_index_cover}~\ref{obs:extend_index_cover_target}, if $\minmat(s)$ has an extend index greater or equal to $j$, then  $\ext(t^*)$ covers $(s,t_j)=(s,t)$. Therefore, the extend index of~$\minmat(s)$ is smaller than $j$. 
Analogously to the proof of part~\ref{lem:no_cover_consequence_source}, this implies that $\pos_{s_j}(\minmat(s)) \geq \pos_{s_j}(t^*)$. This completes the proof as $s_j=\maxmat(t)$.
\end{proof}

Putting the two parts of this lemma together, we immediately obtain the following (cf.~\cref{fig:no_cover}).

\begin{figure}
\centering
\begin{tikzpicture}[scale=0.8, every node/.style={font=\small, circle, draw, fill, inner sep=0pt, minimum size=5pt}]
\node (sstar) at (0,-1) {};
\node (t) at (2, -2) {};
\node (tmax) at (0,-2.7) {};
\node (s) at (0,-4) {};
\node (smin) at (2,-4) {};
\node (tstar) at (2,-5) {};
\draw (tstar) to (tmax) to (smin) to (sstar);
\draw (sstar) to (t);
\draw (s) to (tstar);
\draw[->,thick] (0,-1)++(-10:0.5) arc[start angle=-10,end angle=-80,radius=0.5];
\draw[->,thick] (2,-5)++(110:0.5) arc[start angle=110,end angle=180,radius=0.5];
\draw[<-,thick, red] (0,-2.7)++(-1:0.5) arc[start angle=-1,end angle=-70,radius=0.5];
\draw[->,thick, blue] (2,-4)++(170:0.5) arc[start angle=170,end angle=90,radius=0.5];
\draw[dashed] (s) to node[above, draw=none, fill=none, font=\scriptsize, pos=0.2, gray, anchor=south, yshift=-3pt] {min}(smin);
\draw[dashed] (t) to node[below, draw=none, fill=none, font=\scriptsize, pos=0.2, gray, anchor=north, yshift=3pt, sloped] {max} (tmax);
\node at (-0.2,-1) [draw=none, fill=none, anchor=east,rectangle] {$s^*$};
\node at (-0.2,-2.7) [draw=none, fill=none, anchor=east] {$\maxmat(t)$};
\node at (-0.2,-4) [draw=none, fill=none, anchor=east] {$s$};
\node at (2.2,-4) [draw=none, fill=none, anchor=west] {$\minmat(s)$};
\node at (2.2,-5) [draw=none, fill=none, anchor=west] {$t^*$};
\node at (2.2,-2) [draw=none, fill=none, anchor=west] {$t$};
\node at (-0.1,-3.2) [draw=none,fill=none,anchor=east,red, font=\scriptsize] {by{\hypersetup{linkcolor=red} \cref{lem:no_cover_consequence}~\ref{lem:no_cover_consequence_target}}};
\node at (2.1,-3.5) [draw=none,fill=none,anchor=west,blue, font=\scriptsize] {by {\hypersetup{linkcolor=blue} \cref{lem:no_cover_consequence}~\ref{lem:no_cover_consequence_source}}};
\end{tikzpicture}
\caption{If $\ext(s^*)\cup \ext(t^*)$ does not cover $(s,t)$, then the edges are ordered as indicated.}
\label{fig:no_cover}
\end{figure}

\begin{corollary}
\label{cor:twoside_uncovered}
Let $s,s^*\in S$ and $t, t^* \in T$. If $\ext(s^*)\cup \ext(t^*)$ does \emph{not} cover $(s,t)$, then the walk 
\[
(\minmat(t^*), t^*, \maxmat(t),\minmat(s),s^*)
\]
is temporal.
\end{corollary}

The next lemma connects the coverage of extended stars to cut-crossings in simple stars.

\begin{lemma}
\label{lem:crossing_or_covering}
Let $s^*\in S$ and consider the $s^*$-ordered labeling. Let $k\in \{0,\dots, n-2\}$. Then, for every $i>k$, one of the following holds:
\begin{enumerate}[label=(\roman*), topsep=4pt]
\item $s_i$ is 3-hop-$k$-cut-crossing.
\label{lem:cond_crossing}
\item $\ext(s^*)\cup \ext(t_i)$ covers $(s_j,T)$ for every $j\leq k$.
\label{lem:cond_covering}
\end{enumerate}
\end{lemma}

\begin{proof}
Assume that condition \ref{lem:cond_covering} is not fulfilled, i.e., there exist $j\leq k$ and $t\in T$ such that~$(s_j,t)$ is not covered by $\ext(s^*)\cup \ext(t_i)$.
Then, by \cref{cor:twoside_uncovered}, the walk $(s_i, t_i, \maxmat(t),t_j,s^*)$ is temporal.
This means that~$s_i$ is 3-hop-$k$-cut-crossing (cf.~\cref{fig:no_cover}, \cref{def:cut_crossing}).
\end{proof}

We summarize the findings of this section in the following lemma, which is the key ingredient in the proof of our main result.
It shows that a linear number of edges suffices to cover all source--target pairs in a sufficiently large sub-biclique (satisfying the assumptions of~\cref{lem:recursion}).
Together with dismountability, this will allow us to reduce the problem of covering the remaining pairs to a temporal biclique of roughly half the size.

\begin{lemma}
\label{lem:final}
Fix $s^*\in S$, consider the $s^*$-ordered labeling, and set $k:=\lfloor n/2 \rfloor$.
Let $S':=\{s_0, \dots, s_k\}$ and $T':=\{t_k, \dots, t_{n-1}\}$.
Then there exists a set $E^*$ of at most $6n$ edges containing $\minedges \cup \maxedges$ that covers either $(S,T')$ or $(S',T)$.
\end{lemma}
\begin{proof}
First, assume that every source $s_i$ with $i>k$ is $3$-hop-$k$-cut-crossing,
i.e., the set of 3-hop-$k$-cut-crossing sources is $S_\mathrm{cross}=\{s_{k+1},\ldots,s_{n-1}\}$.
By \cref{obs:cut_crossing}, there exists a set $E'$ of at most $3|S_\mathrm{cross}|\leq 3n$ edges such that $E'\cup \star(s^*)$ covers $(S'\cup S_\mathrm{cross},T')=(S,T')$.
Recall that $\star(s^*)$ contains $\minedges$ and $|\star(s^*)|=2n-1$.
So we can set $E^*:=\star(s^*)\cup E' \cup \maxedges$ to obtain a set of at most $6n$ edges covering $(S,T')$.

Now assume there exists $i>k$ such that $s_i$ is not $3$-hop-$k$-cut-crossing.
Then by \cref{lem:crossing_or_covering}, the set~$E^*:=\ext(s^*)\cup \ext(t_i)$ covers $(S',T)$.
Recall from \cref{def:extended_star} that $\ext(s^*)\cup \ext(t_i)$ consists of $\minedges$, $\maxedges$, all edges incident to $s^*$, all edges incident to $t_i$,
at most one additional edge per source (from $\ext(s^*)$), and at most one additional edge per target (from $\ext(t_i)$).
Each of these six groups contains at most $n$ edges, so $|E^*|\leq 6n$,
and it covers $(S',T)$.
\end{proof}

We now have all the ingredients to prove \cref{thm:mainthm_extremally_matched}.
We restate it here in a slightly stronger form that will allow us to improve the bound on the constant factors for cliques in \cref{sec:improved_factor}.

\begin{theorem}
\label{thm:mainthm_bicliques_stronger}
Every extremally matched temporal bi-clique $(S,T,\lambda)$ with $n:=|S|=|T|$ contains a temporal spanner of size $14n$.
Moreover, this spanner also covers $(S\cup T, S\cup T)$.
\end{theorem}

\begin{proof}
Let $f(n)$ denote the worst-case size of a minimum spanner of an extremally matched temporal bi-clique of size $n$, that is,
\begin{equation*}
f(n):=\max \{\min\{|E'|: E' \text{ is a spanner of } G\}: G=(S,T,\lambda) \text{ is extremally matched with } |S|=n\}.
\end{equation*}
Let $G_n=(S,T,\lambda)$ be an extremally matched temporal bi-clique realizing this maximum.
Fix~${s^*\in S}$, consider the $s^*$-ordered labeling, and let $S',T'$ be as in \cref{lem:final}.
By \cref{lem:final}, there exists a set $E^*$ of at most $6n$ edges covering either $(S,T')$ or $(S',T)$.
In the first case, it remains to find a spanner of $(S,T\setminus T')$, and in the second a spanner of $(S\setminus S',T)$.
In both cases, call this remaining subgraph $G'$.

Note that $|S \setminus S'| = |\{s_{\lfloor n/2 \rfloor +1}, \dots, s_{n-1}\}| \leq \lfloor n/2 \rfloor$, and similarly $|T \setminus T'| = |\{t_0, \ldots, t_{\lfloor n/2 \rfloor -1}\}| = \lfloor n/2 \rfloor$.
Hence, in both cases, $G'$ is a bi-clique with at most $\lfloor 3n/2 \rfloor$ vertices in total, with the smaller side having size at most $\lfloor n/2 \rfloor$ and the larger side having size $n$.
Applying dismountability (i.e., \cref{lem:dismountable}) to $G'$, we obtain $S''\subset S$ and $T''\subset T$ such that $G[S'',T'']$ is extremally matched with~$|S''|=|T''|\leq \lfloor n/2 \rfloor$, and a set of edges $E''$ such that adding $E''$ to any spanner of $G[S'',T'']$ yields a spanner of $G'$, where
\begin{equation*}
|E''|\leq 2\cdot (\lfloor 3n/2 \rfloor - |S''|-|T''|)
\overset{|S''|=|T''|}{\leq}3n - 4|S''|.
\end{equation*}
To summarize, we have selected the edges $E^*\cup E''$ and it only remains to find a spanner of $G[S'',T'']$ to obtain a spanner of $G_n$.
Setting $x_n := |S''| \leq \lfloor n/2 \rfloor$, we have
\begin{equation*}
f(n)\leq |E^*|+|E''|+f(|S''|)
\leq
6n + 3n - 4x_n + f(x_n)
=9n- 4x_n + f(x_n).
\end{equation*}
We now prove by strong induction that $f(n)\leq 14n$. The base case $f(1)=1$ is immediate (and we set $f(0):=0$ for the case where $x_n=0$, i.e., the recursion ends). For the induction step, we have
\begin{equation*}
f(n)\leq 9n-4x_n+f(x_n)
\overset{\text{ind.~hyp.}}{\leq} 9n + (14-4)x_n 
\overset{x_n\leq n/2}{\leq}
 9n + 10 \cdot \frac{n}{2} \leq 14n.
\end{equation*}
Hence spanners of the claimed size exist.

Last, the final edge set contains $\minedges \cup \maxedges$ since this set is contained in the edge set added in the first recursion level chosen according to \cref{lem:final}. This ensures that the spanner does in fact preserve all temporal reachabilities: any target~$t$ reaches the same set of vertices as $\minmat(t)$, since prefixing any temporal walk beginning in $\minmat(t)$ with the edge $\{t, \minmat(t)\}$ yields a valid temporal walk (where we have used that $G$ is extremally matched). Symmetrically, any source~$s$ is reachable from the same set of vertices that reach $\maxmat(s)$, since appending $\{s, \maxmat(s)\}$ to any temporal walk ending at $\maxmat(s)$ again yields a temporal walk.
\end{proof}

Recall from \cref{sec:preliminaries} that \cref{thm:mainthm_extremally_matched} implies~\cref{thm:mainthm_biclique} and that every temporal clique contains a spanner of size $14n$.

\section{A simple improvement on the constant factor for cliques}
\label{sec:improved_factor}

While the focus of our work lies in proving the optimal asymptotics of spanners of temporal cliques, we illustrate here briefly how a result from the literature can be invoked to halve the constant factor for cliques from $14n$ to $7n$.

In \cite{Carnevale25}, the concept of dismountability for cliques was refined as follows.
Given a temporal clique $G=(V,\lambda)$, we say that a vertex $v$ is $\{1,2\}$-hop dismountable if there exist $u,w\in V\setminus \{v\}$ such that $(v,\minmat(u),u)$ is a temporal walk and $(w,\maxmat(w),v)$ is a temporal walk.
If this is the case, one can remove $v$ from the graph, find a spanner of the remaining graph, and then reintroduce~$v$ together with the at most 4 edges of the two walks.
Analogously to the proof of \cref{lem:dismountable}, we obtain the following by repeatedly deleting dismountable vertices.

\begin{observation}
\label{lem:12dismountable}
Let $G=(V,\lambda)$ be a temporal clique. Then there exists $V'\subset V$ and a set $E'$ of at most~$4(|V|-|V'|)$ edges such that no vertex of $G[V']$ is $\{1,2\}$-hop dismountable and, if $E^*$ is a temporal spanner of $G[V']$, then $E^*\cup E'$ is a temporal spanner of $G$.
\end{observation}

We make use of the following known result for which we do not provide a proof here.

\begin{theorem}[Carnevale, Casteigts, Corsini~{\cite[Corollary 4.1]{Carnevale25}}]
\label{thm:dismountablecliques}
Let $G=(V,\lambda)$ be a temporal clique that has no $\{1,2\}$-hop dismountable vertices.
Let $V^-:=\{\minmat(v): v \in V\}$ and $V^+:=\{\maxmat(v): v \in V\}$.
Then $V^-$ and $V^+$ have the same size and form a partition of $V$, and the temporal bi-clique $G':=(V^-,V^+, \lambda)$ is extremally matched.
\end{theorem}

This allows us to halve the factor proven in \cref{obs:biclique_to_clique} and conclude our main result in its final form.
We begin by restating it.

\mainthm*

\begin{proof}
Let $G=(V,\lambda)$ be a temporal clique on $n$ vertices. 
Then by \cref{lem:12dismountable}, there exists~${V'\subseteq V}$ of size $n':=|V'|$, and a set $E'$ of size at most $4(n-n')$ such that $G[V']$ has no~$\{1,2\}$-hop dismountable vertices, and unioning a spanner of $G[V']$ with $E'$ yields a spanner of~$G$.
Then by \cref{thm:dismountablecliques}, the sets~$V^-:=\{\minmat(v): v \in V'\}$ and $V^+:=\{\maxmat(v): v \in V'\}$ (where $\minmat$ and~$\maxmat$ are taken in~$G[V']$) each have size~$n'/2$, partition~$V'$, and the bi-clique $G':=(V^-, V^+,\lambda)$ is extremally matched.
By~\cref{thm:mainthm_bicliques_stronger},~$G'$ contains a spanner $E^*$ of size $14n'/2=7n'$ that also spans $G[V']$.
Together, we obtain a spanner of size
\begin{equation*}
7n' + 4(n-n') = 4n + 3 n' \leq 7n.\qedhere
\end{equation*}
\end{proof}

\section{Recursion in bi-cliques and shifted matching graphs}
\label{sec:shifted_matching_proofs}

In this section, we provide the missing proofs from \cref{sec:intuition}.
We begin by restating the lemma that allows recursion in bi-cliques when a sufficiently large subgraph is already covered. Its proof is analogous to the argument that \cref{lem:final} implies \cref{thm:mainthm_extremally_matched} at the end of \cref{sec:extended_stars} and in fact generalizes it, where we are somewhat less careful about minimizing constant factors.

\recursionlemma*

\begin{proof}
Let $f(n)$ denote the worst-case minimum spanner size over all extremally matched
temporal bi-cliques of size~$n$, i.e.,
\begin{equation*}
  f(n) := \max\bigl\{\min\{|E'| : E' \text{ is a spanner of } G\} :
  G=(S,T,\lambda) \text{ is extremally matched with } |S|=n\bigr\},
\end{equation*}
and let $G=(S,T,\lambda)$~be an extremally matched bi-clique of size~$n$ realizing this
maximum. By assumption, there exist $S'$ and~$T'$ as described, so it remains
to cover $(S \setminus S', T)$ and $(S, T \setminus T')$.

Before proceeding recursively, we apply dismountability (\cref{lem:dismountable}).
Applied to $G[S \setminus S', T]$, \cref{lem:dismountable} yields subsets $S'' \subseteq S \setminus S'$ and $T'' \subseteq T$ such that $G[S'',T'']$ is extremally matched, together with an edge set~$E''$ of size safely upper bounded by $2(|S|+|T|)=4n$ such that any spanner of $G[S'',T'']$ together with~$E''$ forms a spanner of $G[S \setminus S', T]$. Since $S'' \subseteq S \setminus S'$, we have $|S''| \leq n - |S'|$. 
An analogous application of \cref{lem:dismountable} to $G[S, T \setminus T']$ reduces covering that subgraph to an extremally matched bi-clique of size at most $n - |T'|$, at an additional cost of at most $4n$~edges.

To summarize, there exists a set of edges consisting of the $Cn$ edges from \cref{lem:recursion}, and another at most $4n+4n=8n$ edges from applying dismountability, such that the problem reduces to finding spanners of two smaller bi-cliques, one of size $a_n := n - |S'|$ and one of size $b_n := n - |T'|$, which satisfy
$a_n + b_n \leq (1-\delta)n$.
Therefore, we have
\begin{equation*}
  f(n) \leq (C+8)n + f(a_n) + f(b_n).
\end{equation*}
We prove by strong induction that $f(n) \leq \tfrac{C+8}{\delta}\,n$ for all~$n$. 
The base case $f(1) = 1 \leq \tfrac{C+8}{\delta}$ is immediate as $C\geq 1, \delta \in (0,1]$ (and we set $f(0):=0$ for the case where $a_n$ or $b_n$ is 0, i.e., the recursion ends).
For the induction step, we have
\begin{align*}
  f(n)
    &\leq (C+8)\,n + f(a_n) + f(b_n)
    \overset{\text{ind.~hyp.}}{\leq}
     (C+8)\,n + \frac{C+8}{\delta}\,(a_n+b_n) \\
    &\leq (C+8)\,n + \frac{C+8}{\delta}\,(1-\delta)\,n
     = \frac{C+8}{\delta} n. \qedhere
\end{align*}
\end{proof}

Next, we establish the relation between the shifted matching graph and the existence of 1-hop-$k$-cut-crossing sources as stated in \cref{sec:intuition}.
We begin by restating the assertion.

\shiftedmatching*

\begin{figure}
\centering
\begin{tikzpicture}[scale=0.8, every node/.style={font=\small, circle, draw, fill, inner sep=0pt, minimum size=5pt}]
\foreach \i in {0,...,6}{
  \node (s\i) at (0,-\i) {};
\node (t\i) at (3,-\i) {};
}
\draw (s0) to node [above, font=\scriptsize, fill=none, draw=none, pos=0.4] {0}(t0);
\draw (s0) to node [above, font=\scriptsize, fill=none, draw=none, pos=0.4] {1}(t1);
\draw (s0) to node [above, font=\scriptsize, fill=none, draw=none, pos=0.4] {2}(t2);
\draw (s3) to node [above, font=\scriptsize, fill=none, draw=none, pos=0.7] {0}(t3);
\draw (s3) to node [above, font=\scriptsize, fill=none, draw=none, pos=0.7] {1}(t4);
\draw (s3) to node [above, font=\scriptsize, fill=none, draw=none, pos=0.7] {2}(t5);
\draw (s3) to node [above, font=\scriptsize, fill=none, draw=none, pos=0.65, sloped, anchor=south, yshift=-7pt] {$n-1$}(t2);
\draw (s3) to node [above, font=\scriptsize, fill=none, draw=none, pos=0.4, sloped, anchor=south, yshift=-7pt] {$n-2$}(t1);
\draw (t5) to node [above, font=\scriptsize, fill=none, draw=none, pos=0.7] {0}(s5);
\draw (t5) to node [above, font=\scriptsize, fill=none, draw=none, pos=0.7] {1}(s4);
\draw (t5) to node [above, font=\scriptsize, fill=none, draw=none, pos=0.65, sloped, anchor=south, yshift=-7pt] {$n-1$}(s6);
\draw (s3) to node [above, font=\scriptsize, fill=none, draw=none, pos=0.4, sloped, anchor=south, yshift=-7pt] {$n-2$}(t1);
\draw (t6) to node [above, font=\scriptsize, fill=none, draw=none, pos=0.3] {0}(s6);
\draw (s1) to node [below, font=\scriptsize, fill=none, draw=none, pos=0.2] {0}(t1);
\node at (-0.3,0) [draw=none, fill=none, anchor=east] {$s^*=s_0$};
\node at (-0.3,-1) [draw=none, fill=none, anchor=east] {$s_1$};
\node at (-0.3,-2) [draw=none, fill=none, anchor=east] {$s_2$};
\node at (-0.3,-6) [draw=none, fill=none, anchor=east] {$s_{n-1}$};
\node at (3.3,0) [draw=none, fill=none, anchor=west] {$t_0$};
\node at (3.3,-1) [draw=none, fill=none, anchor=west] {$t_1$};
\node at (3.3,-2) [draw=none, fill=none, anchor=west] {$t_2$};
\node at (3.3,-6) [draw=none, fill=none, anchor=west] {$t_{n-1}$};

\draw[->,thick] (10:0.5) arc[start angle=10,end angle=-100,radius=0.5];
\spiralarc[thick]{(0,-3)}{10}{-352}{0.4}{0.65}
\spiralarc[thick]{(3,-5)}{190}{-175}{0.4}{0.65}
\end{tikzpicture}
\hspace{3cm}
\begin{tikzpicture}[scale=0.8, every node/.style={font=\small, circle, draw, fill, inner sep=0pt, minimum size=5pt}]
\node (s0) at (0,0) {};
\node (t0) at (3,0) {};
\draw[red] (s0) to (t0);
\foreach \i in {1,...,6}{
  \node (s\i) at (0,-\i) {};
\node (t\i) at (3,-\i) {};
\draw (s0) to (t\i);
\draw (t0) to (s\i);
}
\foreach \i [evaluate=\i as \j using int(\i-1)] in {1,...,6}{
 \node (s\i) at (0,-\i) {};
\node (t\i) at (3,-\i) {};
\draw (s0) to (t\i);
\draw (t0) to (s\i);
\draw[red] (s\i) to (t\i);
\draw[blue] (s\i) to (t\j);
}
\node at (-0.3,0) [draw=none, fill=none, anchor=east] {$s_0$};
\node at (-0.3,-1) [draw=none, fill=none, anchor=east] {$s_1$};
\node at (-0.3,-2) [draw=none, fill=none, anchor=east] {$s_2$};
\node at (-0.3,-6) [draw=none, fill=none, anchor=east] {$s_{n-1}$};
\node at (3.3,0) [draw=none, fill=none, anchor=west] {$t_0$};
\node at (3.3,-1) [draw=none, fill=none, anchor=west] {$t_1$};
\node at (3.3,-2) [draw=none, fill=none, anchor=west] {$t_2$};
\node at (3.3,-6) [draw=none, fill=none, anchor=west] {$t_{n-1}$};
\draw[->,thick] (10:0.6) arc[start angle=10,end angle=-100,radius=0.6];
\draw[->,thick] (3,0)++(-90:0.7) arc[start angle=-90,end angle=-175,radius=0.7];
\draw[blue] (s0) to (t6);
\end{tikzpicture}
\caption{On the left, a subset of edges of the shifted matching graph is illustrated. On the right, a spanner of the shifted matching graph is depicted. The red edges $\{s_i,t_i\}$ denote the~0-labeled edges. The blue edges $\{s_i, t_{i-1}\}$ (and the edge $\{s_0, t_{n-1}\}$) denote the edges of label $n-1$.}
\label{fig:shifted_matching}
\end{figure}

\begin{proof}
We first show that the shifted matching graph admits no $1$-hop-$k$-cut-crossing sources. 
By symmetry, for any $s^*\in S$, the $s^*$-ordered labeling of the shifted matching graph satisfies again~$\lambda(\{s_i,t_j\})=j-i\bmod n$ (cf.~\cref{fig:shifted_matching}).
Therefore, we use the standard labeling and only have to show that there is no 1-hop-$k$-cut-crossing source for $s_0$. 
For every $0\leq j<i\leq n-1$, we have
\[
  \pos_{t_j}(s_0)=j 
\overset{i\leq n-1}{<} j + (n-i) = n - (i-j)=j-i\bmod n=\pos_{t_j}(s_i).
\]
Observe that this means precisely that there are no $1$-hop-$k$-cut-crossing sources for $s_0$ and any $k$ (cf.~\cref{fig:cut_star}).

For the converse, let $G=(S,T,\lambda)$ be an extremally matched temporal bi-clique with no $1$-hop-$k$-cut-crossing sources. 
The absence of $1$-hop-$k$-cut-crossing sources means precisely (cf.~\cref{fig:cut_star}) that, for every $s_0\in S$, considering the $s_0$-ordered labeling, we have
\begin{equation}\label{eq:no_cut_cross}
  \pos_{t_j}(s_i)>\pos_{t_j}(s_0) \quad\text{for every }0\leq  j<i\leq n-1.
\end{equation}
We derive several consequences. 

First, \eqref{eq:no_cut_cross} implies for each $j$ that, in the ordering of $S$ induced by $t_j$, at least $n-1-j$ vertices appear after $s_0$, and hence $\pos_{t_j}(s_0)\leq n-1 - (n-j-1) =j=\pos_{s_0}(t_j)$.
Since $s_0$ and $t_j$ were arbitrary, we obtain $\pos_t(s)\le\pos_s(t)$ for all $s$ and $t$. 
As $\sum_{s,t}\pos_s(t)=\sum_{s,t}\pos_t(s)=\sum_{i=0}^{n-1}\sum_{j=0}^{n-1} j$, no pair can achieve a strict inequality, so
\begin{equation}\label{eq:pos_coincide}
  \pos_s(t)=\pos_t(s)\quad\text{for all }s\in S,\,t\in T.
\end{equation}
Our findings so far imply the following: Considering the $s_0$-ordered labeling for some $s_0\in S$, for every $j\in\{0,\dots,n-1\}$, in the ordering of $S$ induced by $t_j$, we have by~\eqref{eq:pos_coincide} that~$s_0$ is in $j$-th position and by \eqref{eq:no_cut_cross}  that $\{s_{j+1}, \ldots, s_{n-1}\}$ occupy positions $j+1, \dots, n-1$. 
Therefore, $\{s_0, \ldots, s_j\}$ fill positions $\{0,\ldots, j\}$. By definition, we also have~${\pos_{t_j}(s_j)=0}$. To summarize, we have
\begin{align}
\label{eq:positions}
&\pos_{t_j}(s_j)=0 \nonumber\\ 
&\pos_{t_j}(s_i) \in \{1, \dots, j\} \text{ for } i\in \{0,\dots, j-1\}\\
&\pos_{t_j}(s_i) \in \{j+1, \dots, n-1\} \text{ for } i\in \{j+1,\dots, n-1\}\nonumber
\end{align}
In particular, we have 
\[
\maxmat(t_j)\in\{s_{j+1},\ldots, s_{n-1}\} \text{ for every } j\in \{0,\dots, n-2\}
\]
Applying this at $j=n-2$ gives $\maxmat(t_{n-2})=s_{n-1}$, and since $\maxmat\colon T\to S$ is a bijection, downward induction together with $\maxmat(t_{n-1})=s_0$ yields
\begin{equation}\label{eq:maxmatching}
  \maxmat(t_j)=s_{j+1\bmod n}\quad\text{for all }j.
\end{equation}

Next, let $s\in S$ and $t\in T$ with $t\neq \maxmat(s)$.
In the $s$-ordered labeling, we have $\minmat(\maxmat(s))=s_{n-1}$ and $t=t_j$ for some $j<n-1$.
We thus obtain
\begin{equation}
\label{eq:no_cut_cross_maxmat}
\pos_t(\minmat(\maxmat(s)))=\pos_t(s_{n-1})
\overset{\eqref{eq:no_cut_cross}}{>}
\pos_t(s_0)=\pos_t(s)
\end{equation}

Now, fix $s_0\in S$ and consider the $s_0$-ordered labeling.
Then \eqref{eq:maxmatching} and \eqref{eq:pos_coincide} together imply for every~$i$ that $\minmat(\maxmat(s_i))=s_{i-1 \bmod n}$.
By \eqref{eq:no_cut_cross_maxmat}, this implies that
\begin{equation}
\label{eq:ordering_at_t}
\pos_{t_j}(s_{i-1\bmod n})>\pos_{t_j}(s_i) \text{ for every } j\neq i-1 \bmod n.
\end{equation}

To conclude, observe that \eqref{eq:positions} together with \eqref{eq:ordering_at_t} imply that $\pos_{t_j}(s_i)=j-i\bmod n$. 
With \eqref{eq:pos_coincide}, this also implies that $\pos_{t_j}(s_i)=\pos_{s_i}(t_j)=j-i\bmod n$.
This means precisely that the temporal graph is isomorphic to the shifted matching graph.
\end{proof}

\section{Concluding remarks.}
\label{sec:conclusion}

In this work, we settled the optimal asymptotic size of a minimum spanner of a temporal clique to be $\Theta(n)$, improving on the previously known bound of $\bigO(n \log n)$~\cite{Casteigts21}.

Let us recap the proof of \cref{thm:mainthm_extremally_matched} from an algorithmic perspective. 
Fix an arbitrary $s^*\in S$ and set~$k:=\lfloor n/2 \rfloor$. 
By \cref{lem:crossing_or_covering}, either every source~$s_i$ with $i>k$ is 3-hop-$k$-cut-crossing, or there exists a non-crossing source~$s_i$ such that $\ext(s^*)\cup\ext(t_i)$ covers a large subgraph. 
Checking whether a source is 3-hop-$k$-cut-crossing and finding the corresponding path if so can be done in polynomial time, as can computing $\ext(v)$ for any vertex~$v$. Hence, determining which of the two conditions holds and computing the edge set~$E^*$ in \cref{lem:final} can be carried out in polynomial time.

In the proof that \cref{lem:final} implies \cref{thm:mainthm_extremally_matched},
we apply dismountability and then recurse to a depth of at most~$\log n$, each step taking polynomial time. 
Last, the steps in the proof that \cref{thm:dismountablecliques} implies \cref{thm:mainthm}, including the preliminary $\{1,2\}$-hop dismountability steps from~\cite{Carnevale25}, can also be performed in polynomial time.

\begin{observation}
\leavevmode
\begin{enumerate}[label=(\roman*)]
\item For every temporal bi-clique $(S,T,\lambda)$, a spanner of size
      $10\min(|S|,|T|)+2(|S|+|T|)$ can be computed in polynomial time.
\item For every temporal clique on~$n$ vertices, a spanner of size~$7n$
      can be computed in polynomial time.
\end{enumerate}
\end{observation}

With the asymptotics now settled, the natural next question is to determine the correct constants.
Our aim in this work was to establish a linear bound, and we make no claim that our bound of $7n$ is anywhere near optimal.
We believe that optimizing it further lies beyond the scope of this work, but we regard it as an interesting direction for future research.
The best known lower bound is $2n-4$~\cite{bumby1981}.
Moreover, the authors of~\cite{Casteigts21} report that, despite an extensive computer search, they found no instance failing to admit a spanner with $2n-3$ edges, and they pose the following open problem.
\begin{question}
    Does every temporal clique admit a spanner of size $2n-3$?
\end{question}

Recent research~\cite{BaligacsBBDZ26,BiloESA22} has focused not only on minimizing the size of a temporal spanner, but also its \emph{stretch}, that is, the maximum length of a temporal reachability path in the spanner.
We remark that, in \cref{lem:final}, the set $E^*$ in fact covers the claimed subgraphs in such a way that each reachability path has length at most~5, so one might hope that our techniques carry over to this setting.
This is not the case, however.
The main obstacle is our heavy reliance on dismountability (\cref{lem:dismountable}): we sometimes remove a vertex, find a spanner of the remaining graph, and then reintroduce the vertex, which can add~2 to the stretch of the spanner.
This paradigm is therefore incompatible with bounding the stretch.
In particular, \cref{lem:final} assumes the bi-clique to be extremally matched, and this assumption is without loss of generality only when we are allowed to use dismountability.
While~\cite{BiloESA22} shows that temporal cliques need not admit sparse spanners of stretch~2, recent work~\cite{BaligacsBBDZ26} establishes that every temporal clique admits a spanner of stretch~3 and size~$n^{1+o(1)}$.
This makes the following two questions about spanners of linear size and constant stretch particularly intriguing.
\begin{question}
    Does every temporal clique admit a spanner of size $\bigO(n)$ and stretch $\bigO(1)$?
\end{question}
\begin{question}
    Does every temporal clique admit a spanner of size $\bigO(n)$ and stretch~3?
\end{question}

\paragraph*{Acknowledgments.}
I am grateful to Davide Bilò for suggesting the related question of temporal spanners with bounded stretch at the AlgUW workshop in Będlewo (September 2025), which first drew my attention to the problem resolved in the present paper. I thank him, along with Václav Blažej, Maël Dumas, and Anna Zych-Pawlewicz, for the fruitful discussions that followed.
These discussions led to our joint work~\cite{BaligacsBBDZ26}, which in turn inspired the present paper.

\bibliographystyle{alpha}
\bibliography{LinearSpannerReferences}

\end{document}